\documentclass[12pt,fleqn]{article}
\usepackage[dvips]{epsfig}
\usepackage{amsmath,amssymb,amsfonts}
\usepackage{tikz}
\usepackage{tikz,tikz-cd}
\usepackage{hyperref} 

\numberwithin{equation}{subsection}
\numberwithin{figure}{subsection}

\newtheorem{theorem}{Theorem}[subsection]

\newtheorem{remark}[theorem]{Remark}

\newtheorem{proposition}[theorem]{Proposition}

\def\pr{\operatorname{pr}}%
\def\spec{\operatorname{Spec}}%
\def\tr{\operatorname{tr}}%

\newenvironment{proof}[1][Proof:]{\begin{trivlist}
\item[\hskip \labelsep {\bfseries #1}]}{\end{trivlist}}

\newcommand{\qed}{\nobreak \ifvmode \relax \else
      \ifdim\lastskip<1.5em \hskip-\lastskip
      \hskip1.5em plus0em minus0.5em \fi \nobreak
      \vrule height0.75em width0.5em depth0.25em\fi}

 0 3  
\usepackage{hyperref} 

\date{}
\begin{document}

\title{On the Protection Against Noise for Measurement-Based Quantum Computation}

\author{Valentin Vankov Iliev\\
Institute of Mathematics and Informatics\\
Bulgarian Academy of Sciences\\
Sofia, Bulgaria\\
e-mail: viliev@math.bas.bg\\
 }

\maketitle

\begin{quote}

"In principio erat verbum,..."

\emph{Ioann 1:1}

"It is operationally impossible to separate reality and
information."

\emph{Anton Zeilinger}

\end{quote}

\begin{abstract}
Here we establish conditions for some pairs of quantum logic gates
which operate on one qubit to be protected against crosstalk.

\end{abstract}

\section{Introduction, Notation}

\label{1}

\subsection{Notation}

\label{1.1}

The following notation will be frequently used in this paper:

\noindent $\delta_{k,\ell}$: Kronecker's delta;

\noindent $\mathbb{Z}_2=\{0,1\}$: the additive group of two
elements;

\noindent $\mathcal{H}$: $2$-dimensional unitary space with inner
product $\langle u|v\rangle$ of vectors $|u\rangle$ and
$|v\rangle$;

\noindent $\{|0\rangle,|1\rangle\}$: orthonormal frame for
$\mathcal{H}$, sometimes called the computational basis;

\noindent $\mathbb{I}_{\mathcal{H}}$: the identity linear operator
on $\mathcal{H}$;

\noindent $\spec(A)$: the spectre of a linear operator $A$ on
$\mathcal{H}$;
\[
\!\sigma_1=\left(
\begin{array}{ccccc}
0 &1\\
1 & 0 \\
\end{array}
\right), \sigma_2=\left(
\begin{array}{ccccc}
0 &-\mathrm{i}\\
\mathrm{i} & 0 \\
\end{array}
\right), \sigma_3=\left(
\begin{array}{ccccc}
1 &0 \\
0 & -1 \\
 \end{array}
\right), H=\frac{1}{\sqrt{2}}\left(
\begin{array}{ccccc}
1 &1 \\
1 & -1 \\
 \end{array}
\right):
\]
Pauli matrices and Hadamard matrix, respectively;

\noindent $\mathcal{E}$: the $3$-dimensional real linear space of
all self-adjoined operators on $\mathcal{H}$ with trace $0$, which
is furnished with orthonormal coordinates $x$, $y$, $z$;

\noindent $\mathbb{R}_+(A)=\{[(cx,cy,cz)|c\in\mathbb{R},c\geq
0\}$: direction of the operator $A\in\mathcal{E}$ with coordinates
$(x,y,z)$;

\noindent $\mathcal{H}^{\otimes 2}$: the tensor square of
$\mathcal{H}$ provided with the standard structure of a unitary
space, the tensor product of vectors $|u\rangle$ and $|v\rangle$
being denoted $|uv\rangle$;

\noindent $\mathcal{U}^{\left(2\right)}$: the unit sphere of
$\mathcal{H}^{\otimes 2}$;

\noindent $\psi_{s,t}=
\frac{1}{\sqrt{2}}(|0t\rangle+(-1)^s|1(t+1)\rangle)$,
$\psi_{s,t}\in\mathcal{U}^{\left(2\right)}$, $s,t\in\mathbb{Z}_2$:
the four Bell's states;

\noindent We denote: $\tr_0\alpha=\cos\alpha$,
$\tr_1\alpha=\sin\alpha$, $\alpha\in\mathbb{R}$.

\noindent We denote:
$P(\alpha,\beta)=\tr_0\alpha\tr_0\beta\tr_1\alpha\tr_1\beta$,
$\alpha,\beta\in\mathbb{R}$.

\subsection{Introduction}

\label{1.5}

Some of the logic gates in a quantum computational network can be
represented by self-adjoined operators $A$, $B$, $\ldots$ with
spectre $\{1,-1\}$ on the unitary plane $\mathcal{H}$. Their set
coincides with unit sphere of the linear space $\mathcal{E}$.
After appropriate tensoring with the identity operator, any pair
$A$, $B$ of such logic operations generates a pair $\mathcal{A}$
and $\mathcal{B}$ of self-adjoined commuting operators with the
same spectre, which are defined on the bipartite quantum system
$\mathcal{H}^{\otimes 2}$. In particular, the measurements of the
observables corresponding to these operators can be considered as
simultaneous. In other words, we consider as an event "the outcome
of measuring $\mathcal{A}$ is $\lambda_k$ \emph{and} the outcome
of measuring $\mathcal{B}$ is $\lambda_\ell$ ", where $\lambda_k,
\lambda_\ell\in\{1,-1\}$. We can suppose without loss of
generality that $k,\ell\in\mathbb{Z}_2$, and that $\lambda_0=1$,
$\lambda_1=-1$. Born's rule allows us to interpret the above
conjunction as an intersection of two events in a sample space and
this is done in Section~\ref{5}. Namely, the tensor product of
orthonormal frames of the unitary plane, consisting of
eigenvectors of $A$ and $B$ constitutes an orthonormal frame of
$\mathcal{H}^{\otimes 2}$ whose members are common eigenvectors of
both $\mathcal{A}$ and $\mathcal{B}$. We chose this frame as the
set of outcomes of a sample space
$S(\psi;\mathcal{A},\mathcal{B})$ with probability assignment
$p_{k,\ell}$ created by quantum theory via fixing a state
$\psi\in\mathcal{H}^{\otimes 2}$ and via Born's rule. It turns out
that $p_{k,\ell}$ is the probability of the intersection of the
events $\mathcal{A}=\lambda_k$ and $\mathcal{B}=\lambda_\ell$ in
this sample space. In order to find the probability $p_{k,\ell}$,
we identify the self-adjoined operators with their matrices with
respect to the computational basis for $\mathcal{H}$ and express
the matrices $A$ and $B$ by using polar coordinates,
$A=A_{\mu,\eta}$ and $B=A_{\nu,\zeta}$, see Theorem~\ref{5.1.5}.
Here $\mu$, $\nu$ are polar angles and $\eta$, $\zeta$ are
azimuthal angles.

Under the condition $\psi=\psi_{s,t}$ for some
$s,t\in\mathbb{Z}_2$, (one of the four Bell's states)
Theorem~\ref{5.5.1} yields that the probabilities of all events
$\mathcal{A}=\lambda_k$ and $\mathcal{B}=\lambda_\ell$,
$k,\ell\in\mathbb{Z}_2$, are equal and that the quantity
$p_{k,\ell}$ is invariant with respect to the natural action of
the group $\mathbb{Z}_2$ onto the Klein four-group
$\mathbb{Z}_2\times \mathbb{Z}_2$. In particular, the probability
assignment of the sample space
$S(\psi_{s,t};\mathcal{A},\mathcal{B})$ can be written in the form
$(p_{k,k},p_{k,k+1}, p_{k,k+1}, p_{k,k})$ for $k\in\mathbb{Z}_2$.

In Section~\ref{10} we consider the two binary trials
$\mathfrak{A}=(\mathcal{A}=1)\cup(\mathcal{A}=-1)$,
$\mathfrak{B}=(\mathcal{B}=1)\cup(\mathcal{B}=-1)$ and make use of
the average quantity of information of one of the experiments
$\mathfrak{A}$ and $\mathfrak{B}$ relative to the other (the
\emph{information flow}, or, \emph{noise}, or,  \emph{crosstalk}
between them), see~\cite[\S 1]{[1]}, defined in this particular
case by the Shannon's formula~\cite[5.3, (6)]{[5]}.

In accord with~\cite{[5]}, the joint experiment of the above
binary trials gives rise to the same probability distribution
$(p_{k,k},p_{k,k+1}, p_{k,k+1}, p_{k,k})$, $k\in\mathbb{Z}_2$. Via
modification of its entropy, we bring out a function which
measures the degree of dependence of the events
$\mathcal{A}=\lambda_k$ and $\mathcal{B}=\lambda_\ell$. The above
events are independent (the entropy is maximal) if and only if the
crosstalk between $\mathfrak{A}$ and $\mathfrak{B}$ is zero. In
this case we also say that the measurements of the observables
corresponding to $\mathcal{A}$ and $\mathcal{B}$ are
\emph{informationally independent}. Thus,
following~\cite[5.2]{[5]}, we conclude that the above events are
informationally independent exactly when the equation
$p_{k,k}=\frac{1}{4}$ is satisfied.

Section~\ref{15} is devoted to three particular cases when the
expression for $p_{k,k}$ from Theorem~\ref{5.5.1}, {\rm (ii)},
{\rm (iii)}, has very simple form and the above equation can be
solved explicitly in terms of sums or differences of polar or
azimuthal angles. Those cases are defined by the property that
both operators $A$, $B$ have simultaneously directions
$\mathbb{R}_+(A)$, $\mathbb{R}_+(B)$ laying in one of the three
coordinate planes $x=0$, $y=0$, $z=0$ of $\mathcal{E}$. In
Subsection~\ref{15.1} we consider the set of operators with
direction in coordinate plane $x=0$ (a circle on the unit sphere
of $\mathcal{E}$). Pauli matrices $\sigma_2$ and $\sigma_3$ belong
to this set. The results of measurements performed on the
observables $\mathcal{A}_{\mu,\eta}$ and $\mathcal{B}_{\nu,\zeta}$
are informationally independent if and only if
$\mu+\nu=\frac{\pi}{2}$ or $\mu+\nu=\frac{3\pi}{2}$ in case
$\psi=\frac{1}{\sqrt{2}}(|0t\rangle+|1(t+1)\rangle)$ for some
$t\in\mathbb{Z}_2$, and if and only if $|\mu-\nu|=\frac{\pi}{2}$
in case $\psi=\frac{1}{\sqrt{2}}(|0t\rangle-|1(t+1)\rangle)$ for
some $t\in\mathbb{Z}_2$.

We study the set of operators with direction in coordinate plane
$y=0$ in Subsection~\ref{15.5}. Pauli matrices $\sigma_1$ and
$\sigma_3$, as well as Hadamard matrix $H$, are members of this
set. It turns out that the outcomes of measurement of observables
$\mathcal{A}_{\mu,\eta}$ and $\mathcal{B}_{\nu,\zeta}$ are
informationally independent exactly in case
$\mu+\nu=\frac{\pi}{2}$ or $\mu+\nu=\frac{3\pi}{2}$, when
$\psi=\frac{1}{\sqrt{2}}(|0(s+1)\rangle+(-1)^s|1s\rangle)$ for
some $s\in\mathbb{Z}_2$, and in case $|\mu-\nu|=\frac{\pi}{2}$,
when $\psi=\frac{1}{\sqrt{2}}(|0s\rangle+(-1)^s|1(s+1)\rangle)$
for some $s\in\mathbb{Z}_2$.

Finally, in Subsection~\ref{15.10} we examine the set of operators
with direction in coordinate plane $z=0$. Pauli operators
$\sigma_1$ and $\sigma_2$ belong there. In this case the
measurements performed by the observables $\mathcal{A}_{\mu,\eta}$
and $\mathcal{B}_{\nu,\zeta}$ are informationally independent
exactly in case
$\eta+\zeta\in\{\frac{\pi}{2},\frac{3\pi}{2},\frac{5\pi}{2},\frac{7\pi}{2}$\},
when $\psi=\frac{1}{\sqrt{2}}(|0,0\rangle+(-1)^s|1,1\rangle)$ for
some $s\in\mathbb{Z}_2$, and in case $|\eta-\zeta|=\frac{\pi}{2}$
or $|\eta-\zeta|=\frac{3\pi}{2}$, when
$\psi=\frac{1}{\sqrt{2}}(|0,1\rangle+(-1)^s|1,0\rangle)$ for some
$s\in\mathbb{Z}_2$.

When all was said and done, we hope that the above conditions for
absence of crosstalk between two interacting logic gates can be
checked experimentally.

\section{Two Commuting Operators on $\mathcal{H}^{\otimes 2}$}

\label{5}

The self-adjoined operators on $\mathcal{H}$ with spectre
$\{1,-1\}$ have, in general, the form
\[
A_{\mu,\eta}=\left(
\begin{array}{ccccc}
\cos\mu &e^{-\mathrm{i}\eta}\sin\mu \\
e^{\mathrm{i}\eta}\sin\mu & -\cos\mu \\
 \end{array}
\right),
\]
where $\mu\in [0,\pi]$ is the polar angle and $\eta\in [0,2\pi)$
is the azimuthal angle. The polar coordinates $x=\sin\mu\cos\eta$,
$y=\sin\mu\sin\eta$, $z=\cos\mu$, establish an isomorphism of the
set of above operators and the unit sphere in $\mathcal{E}$.

We denote $B_{\nu,\zeta}=A_{\nu,\zeta}$, $\lambda_0=1$,
$\lambda_1=-1$. Corresponding (normalized) eigenvectors are
\[
u_{\mu,\eta}^{\left(k\right)}=
(-1)^ke^{-\mathrm{i}\eta}\tr_k\frac{\mu}{2}|0\rangle+\tr_{k+1}\frac{\mu}{2}|1\rangle,
k\in\mathbb{Z}_2.
\]
Moreover, $H_{\mu,\eta}^{\left(k\right)}=\mathbb{C}
u_{\mu,\eta}^{\left(k\right)}$ is its $\lambda_k$-eigenspace. Note
that
\[
 \{|00\rangle, |01\rangle, |10\rangle, |11\rangle\} \hbox{\rm\
and\ }u_{\mu,\eta}u_{\nu,\zeta} = \{|u_{\mu,\eta}^{\left(0\right)}
u_{\nu,\zeta}^{\left(0\right)}\rangle,
|u_{\mu,\eta}^{\left(0\right)}
u_{\nu,\zeta}^{\left(1\right)}\rangle,
|u_{\mu,\eta}^{\left(1\right)}
u_{\nu,\zeta}^{\left(0\right)}\rangle,
|u_{\mu,\eta}^{\left(1\right)}
u_{\nu,\zeta}^{\left(1\right)}\rangle\}
\]
are orthonormal frames for $\mathcal{H}^{\otimes 2}$.

Let us set
$\mathcal{A}_{\mu,\eta}=A_{\mu,\eta}\otimes\mathbb{I}_{\mathcal{H}}$,
$\mathcal{B}_{\nu,\zeta}=\mathbb{I}_{\mathcal{H}}\otimes
B_{\nu,\zeta}$. It is a straightforward check that the last two
linear operators on $\mathcal{H}^{\otimes 2}$ are also
self-adjoined with spectre $\{1,-1\}$. Moreover, the
$\lambda_k$-eigenspace
$\mathcal{H}_{\mathcal{A}_{\mu,\eta}}^{\left(k\right)}=
H_{\mu,\eta}^{\left(k\right)}\otimes\mathcal{H}$ of the operator
$\mathcal{A}_{\mu,\eta}$ has orthonormal frame
$\{|u_{\mu,\eta}^{\left(k\right)}
u_{\nu,\zeta}^{\left(0\right)}\rangle,|u_{\mu,\eta}^{\left(k\right)}
u_{\nu,\zeta}^{\left(1\right)}\rangle\}$, and the
$\lambda_\ell$-eigenspace
$\mathcal{H}_{\mathcal{B}_{\nu,\zeta}}^{\left(\ell\right)}=
\mathcal{H}\otimes H_{\nu,\zeta}^{\left(\ell\right)}$ of the
operator $\mathcal{B}_{\nu,\zeta}$ has orthonormal frame
$\{|u_{\mu,\eta}^{\left(0\right)}
u_{\nu,\zeta}^{\left(\ell\right)}\rangle,
|u_{\mu,\eta}^{\left(1\right)}
u_{\nu,\zeta}^{\left(\ell\right)}\rangle\}$,
$k,\ell\in\mathbb{Z}_2$.

Since $u_{\mu,\eta}u_{\nu,\zeta}$ is an orthonornal frame of
$\mathcal{H}^{\otimes 2}$ consisting of eigenvectors of both
$\mathcal{A}_{\mu,\eta}$ and $\mathcal{B}_{\nu,\zeta}$, then the
last two operators commute. In other words, the outcomes of
measurements of these observables can be thought as simultaneous
(the order is irrelevant) and the quantum theory predicts the
probabilities of both outcomes by employing the frame
$u_{\mu,\eta}u_{\nu,\zeta}$. Formally, this is done below.

\subsection{A Sample Space and Two Random Variables}

\label{5.1}

Let $\psi\in\mathcal{U}^{\left(2\right)}$ and let
$S(\psi;\mathcal{A},\mathcal{B})$ be the sample space with set of
outcomes $u_{\mu,\eta}u_{\nu,\zeta}$ and probability assignment
$p_{k,\ell}=|\langle|u_{\mu,\eta}^{\left(k\right)}
u_{\nu,\zeta}^{\left(\ell\right)}\rangle|\psi\rangle|^2$,
$k,\ell\in\mathbb{Z}_2$.

With an abuse of the language, we consider the observable
$\mathcal{A}_{\mu,\eta}$ as a random variable
\[
\mathcal{A}_{\mu,\eta}\colon
u_{\mu,\eta}u_{\nu,\zeta}\to\mathbb{R},
\mathcal{A}_{\mu,\eta}(|u_{\mu,\eta}^{\left(0\right)}
u_{\nu,\zeta}^{\left(\ell\right)}\rangle)=1,
\mathcal{A}_{\mu,\eta}(|u_{\mu,\eta}^{\left(1\right)}
u_{\nu,\zeta}^{\left(\ell\right)}\rangle)=-1, \ell\in\mathbb{Z}_2,
\]
on the sample space $S(\psi;\mathcal{A},\mathcal{B})$ with
probability distribution
\[
p_{\mathcal{A}_{\mu,\eta}}(\lambda_k)=p_{k,0}+p_{k,1},
k\in\mathbb{Z}_2, p_{\mathcal{A}_{\mu,\eta}}(\lambda)=0,
\lambda\notin\spec(\mathcal{A}_{\mu,\eta}).
\]
Identifying the event $\{|u_{\mu,\eta}^{\left(k\right)}
u_{\nu,\zeta}^{\left(0\right)}\rangle,|u_{\mu,\eta}^{\left(k\right)}
u_{\nu,\zeta}^{\left(1\right)}\rangle\}$ with the "event"
$\mathcal{A}_{\mu,\eta}=\lambda_k$ (the result of the
measurement), we have
$\pr(\mathcal{A}_{\mu,\eta}=\lambda_k)=p_{k,0}+p_{k,1}$,
$k\in\mathbb{Z}_2$.

We also consider the observable $\mathcal{B}_{\nu,\zeta}$ as a
random variable
\[
\mathcal{B}_{\nu,\zeta}\colon
u_{\mu,\eta}u_{\nu,\zeta}\to\mathbb{R},
\mathcal{B}_{\nu,\zeta}(|u_{\mu,\eta}^{\left(k\right)}
u_{\nu,\zeta}^{\left(0\right)}\rangle)=1,
\mathcal{B}_{\nu,\zeta}(|u_{\mu,\eta}^{\left(k\right)}
u_{\nu,\zeta}^{\left(1\right)}\rangle)=-1, k\in\mathbb{Z}_2,
\]
on the above sample space with probability distribution
\[
p_{\mathcal{B}_{\nu,\zeta}}(\lambda_\ell)=p_{0,\ell}+p_{1,\ell},
\ell\in\mathbb{Z}_2, p_{\mathcal{B}_{\nu,\zeta}}(\lambda)=0,
\lambda\notin\spec(\mathcal{B}_{\nu,\zeta}).
\]
Identifying the event $\{|u_{\mu,\eta}^{\left(0\right)}
u_{\nu,\zeta}^{\left(\ell\right)}\rangle,|u_{\mu,\eta}^{\left(1\right)}
u_{\nu,\zeta}^{\left(\ell\right)}\rangle\}$ with the "event"
$\mathcal{B}_{\nu,\zeta}=\lambda_\ell$, we have
$\pr(\mathcal{B}_{\nu,\zeta}=\lambda_\ell)=p_{0,\ell}+p_{1,\ell}$,
$\ell\in\mathbb{Z}_2$. Moreover, the equality
$(\mathcal{A}_{\mu,\eta}=\lambda_k)\cap(\mathcal{B}_{\nu,\zeta}=\lambda_\ell)=
\{|u_{\mu,\eta}^{\left(k\right)}
u_{\nu,\zeta}^{\left(\ell\right)}\rangle\}$ yields
$\pr((\mathcal{A}_{\mu,\eta}=\lambda_k)\cap(\mathcal{B}_{\nu,\zeta}=\lambda_\ell))=
p_{k,\ell}$,  $k,\ell\in\mathbb{Z}_2$.

We obtain immediately

\begin{proposition} \label{5.1.1} The following two statements are
equivalent:

({\rm i}) One has $p_{0,0}=p_{1,1}$ and $p_{0,1}=p_{1,0}$.

({\rm ii}) For all $k,\ell\in\mathbb{Z}_2$ one has
$\pr(\mathcal{A}_{\mu,\eta}=\lambda_k)=\pr(\mathcal{B}_{\nu,\zeta}=\lambda_\ell)$.

\end{proposition}

\begin{theorem}\label{5.1.5} Let $\psi=\psi_{s,t}$ where $s,t\in\mathbb{Z}_2$.
Then for all $k,\ell\in\mathbb{Z}_2$ one has
\[
p_{k,\ell}=\frac{1}{2}|
(-1)^{k+\ell}e^{-\mathrm{i}\left(\eta+\zeta\right)}\tr_k\frac{\mu}{2}
\tr_\ell\frac{\nu}{2}\delta_{0,t}+
(-1)^ke^{-\mathrm{i}\eta}\tr_k\frac{\mu}{2}\tr_{\ell+1}\frac{\nu}{2}\delta_{1,t}+
\]
\[
(-1)^{s+\ell}e^{-\mathrm{i}\zeta}\tr_{k+1}\frac{\mu}{2}
\tr_\ell\frac{\nu}{2}\delta_{0,t+1}+
(-1)^s\tr_{k+1}\frac{\mu}{2}\tr_{\ell+1}\frac{\nu}{2}\delta_{1,t+1}|^2.
\]

\end{theorem}

\begin{proof}  We have
\[
p_{k,\ell}=|\langle u_{\mu,\eta}^{\left(k\right)}
u_{\nu,\zeta}^{\left(\ell\right)}\rangle|\psi_{s,t}\rangle|^2=
\frac{1}{2}|\langle u_{\mu,\eta}^{\left(k\right)}|0\rangle \langle
u_{\nu,\zeta}^{\left(\ell\right)}|t\rangle+ (-1)^s\langle
u_{\mu,\eta}^{\left(k\right)}|1\rangle \langle
u_{\nu,\zeta}^{\left(\ell\right)}|t+1\rangle|^2=
\]
\[
\frac{1}{2}|(
(-1)^ke^{-\mathrm{i}\eta}\tr_k\frac{\mu}{2}\delta_{0,0}+\tr_{k+1}\frac{\mu}{2}\delta_{1,0})
((-1)^\ell
e^{-\mathrm{i}\zeta}\tr_\ell\frac{\nu}{2}\delta_{0,t}+\tr_{\ell+1}\frac{\nu}{2}\delta_{1,t})+
\]
\[
(-1)^s(
(-1)^ke^{-\mathrm{i}\eta}\tr_k\frac{\mu}{2}\delta_{0,1}+\tr_{k+1}\frac{\mu}{2}\delta_{1,1})
((-1)^\ell
e^{-\mathrm{i}\zeta}\tr_\ell\frac{\nu}{2}\delta_{0,t+1}+\tr_{\ell+1}\frac{\nu}{2}\delta_{1,t+1})|^2=
\]
\[
\frac{1}{2}|(-1)^ke^{-\mathrm{i}\eta}\tr_k\frac{\mu}{2} ((-1)^\ell
e^{-\mathrm{i}\zeta}\tr_\ell\frac{\nu}{2}\delta_{0,t}+\tr_{\ell+1}\frac{\nu}{2}\delta_{1,t})+
\]
\[
(-1)^s\tr_{k+1}\frac{\mu}{2}((-1)^\ell
e^{-\mathrm{i}\zeta}\tr_\ell\frac{\nu}{2}\delta_{0,t+1}+\tr_{\ell+1}\frac{\nu}{2}\delta_{1,t+1})|^2=
\]
\[
\frac{1}{2}|
(-1)^{k+\ell}e^{-\mathrm{i}\left(\eta+\zeta\right)}\tr_k\frac{\mu}{2}
\tr_\ell\frac{\nu}{2}\delta_{0,t}+
(-1)^ke^{-\mathrm{i}\eta}\tr_k\frac{\mu}{2}\tr_{\ell+1}\frac{\nu}{2}\delta_{1,t}+
\]
\[
(-1)^{s+\ell}e^{-\mathrm{i}\zeta}\tr_{k+1}\frac{\mu}{2}
\tr_\ell\frac{\nu}{2}\delta_{0,t+1}+
(-1)^s\tr_{k+1}\frac{\mu}{2}\tr_{\ell+1}\frac{\nu}{2}\delta_{1,t+1}|^2.
\]

\end{proof}

\subsection{The Probability Assignment of the Sample Space
$S(\psi;\mathcal{A},\mathcal{B})$}

\label{5.5}

Below we express the probabilities $p_{k,\ell}$,
$k,\ell\in\mathbb{Z}_2$, as functions in sums $\mu+(-1)^s\nu$ and
$\eta+(-1)^t\zeta$, $s,t\in\mathbb{Z}_2$, by using
Theorem~\ref{5.1.5}.

\begin{theorem}\label{5.5.1} Let $\psi=\psi_{s,t}$ for some
$s,t\in\mathbb{Z}_2$.

{\rm (i)}  One has
$\pr(\mathcal{A}_{\mu,\eta}=\lambda_k)=\pr(\mathcal{B}_{\nu,\zeta}=\lambda_\ell)=
\frac{1}{2}$ for any $k,\ell\in\mathbb{Z}_2$.

{\rm (ii)} One has
\[
p_{0,0}=p_{1,1}=
\]
\[
\frac{1}{2}\tr_t^2\frac{\mu+(-1)^s\nu}{2}+2(-1)^{s+t}\tr_t^2\frac{\eta+(-1)^t\zeta}{2}P\left(\frac{\mu}{2},
\frac{\nu}{2}\right)=
\]
\[
\frac{1}{2}\tr_t^2\frac{\mu+(-1)^{s+1}\nu}{2}+
2(-1)^{s+t+1}\tr_{t+1}^2\frac{\eta+(-1)^t\zeta}{2}P\left(\frac{\mu}{2},
\frac{\nu}{2}\right)
\]
for any $s,t\in\mathbb{Z}_2$.

{\rm (iii)} One has
\[
p_{0,1}=p_{1,0}=
\]
\[
\frac{1}{2}\tr_{t+1}^2\frac{\mu+(-1)^s\nu}{2}+
2(-1)^{s+t+1}\tr_t^2\frac{\eta+(-1)^t\zeta}{2}P\left(\frac{\mu}{2},
\frac{\nu}{2}\right)=
\]
\[
\frac{1}{2}\tr_{t+1}^2\frac{\mu+(-1)^{s+1}\nu}{2}+
2(-1)^{s+t}\tr_{t+1}^2\frac{\eta+(-1)^t\zeta}{2}P\left(\frac{\mu}{2},
\frac{\nu}{2}\right)
\]
for any $s,t\in\mathbb{Z}_2$.

\end{theorem}

\begin{proof} {\rm (i)} In case $t=0$ Theorem~\ref{5.1.5} yields
\[
p_{0,0}=\frac{1}{2}|
e^{-\mathrm{i}\left(\eta+\zeta\right)}\tr_0\frac{\mu}{2}
\tr_0\frac{\nu}{2}+(-1)^s\tr_1\frac{\mu}{2}\tr_1\frac{\nu}{2}|^2,
\]
\[
p_{ 1,1}=\frac{1}{2}|
e^{\mathrm{i}\left(\eta+\zeta\right)}\tr_0\frac{\mu}{2}
\tr_0\frac{\nu}{2}+(-1)^s\tr_1\frac{\mu}{2}\tr_1\frac{\nu}{2}|^2,
\]
and hence $p_{0,0}=p_{1,1}$. Similarly, we obtain
\[
p_{0,1}=\frac{1}{2}|
e^{-\mathrm{i}\left(\eta+\zeta\right)}\tr_0\frac{\mu}{2}
\tr_1\frac{\nu}{2}+
(-1)^{s+1}\tr_1\frac{\mu}{2}\tr_0\frac{\nu}{2}|^2,
\]
\[
p_{1,0}=\frac{1}{2}|
e^{\mathrm{i}\left(\eta+\zeta\right)}\tr_0\frac{\mu}{2}
\tr_1\frac{\nu}{2}+
(-1)^{s+1}\tr_1\frac{\mu}{2}\tr_0\frac{\nu}{2}|^2,
\]
and therefore $p_{0,1}=p_{1,0}$.

In case $t=1$ Theorem~\ref{5.1.5} implies
\[
p_{0,0}=\frac{1}{2}|
e^{-\mathrm{i}\left(\eta-\zeta\right)}\tr_0\frac{\mu}{2}\tr_1\frac{\nu}{2}+
(-1)^s\tr_1\frac{\mu}{2}\tr_0\frac{\nu}{2}|^2,
\]
\[
p_{1,1}=\frac{1}{2}|
e^{\mathrm{i}\left(\eta-\zeta\right)}\tr_0\frac{\mu}{2}\tr_1\frac{\nu}{2}+
(-1)^s\tr_1\frac{\mu}{2}\tr_0\frac{\nu}{2}|^2,
\]
and therefore $p_{0,0}=p_{1,1}$. Similarly, we have
\[
p_{0,1}=\frac{1}{2}|
e^{-\mathrm{i}\left(\eta-\zeta\right)}\tr_0\frac{\mu}{2}\tr_0\frac{\nu}{2}+
(-1)^{s+1}\tr_1\frac{\mu}{2}\tr_1\frac{\nu}{2}|^2,
\]
\[
p_{1,0}=\frac{1}{2}|
e^{\mathrm{i}\left(\eta-\zeta\right)}\tr_0\frac{\mu}{2}\tr_0\frac{\nu}{2}+
(-1)^{s+1}\tr_1\frac{\mu}{2}\tr_1\frac{\nu}{2}|^2,
\]
so $p_{0,1}=p_{1,0}$.

In accord with Proposition~\ref{5.1.1}, for any
$k,\ell\in\mathbb{Z}_2$ we have
$\pr(\mathcal{A}_{\mu,\eta}=\lambda_k)=\pr(\mathcal{B}_{\nu,\zeta}=\lambda_\ell)$
and part {\rm (i)} is proved.

{\rm (ii)} In case $t=0$ we have
\[
p_{0,0}=\frac{1}{2}|
(\tr_0\left(\eta+\zeta\right)-\mathrm{i}\tr_1\left(\eta+\zeta\right))\tr_0\frac{\mu}{2}
\tr_0\frac{\nu}{2}+(-1)^s\tr_1\frac{\mu}{2}\tr_1\frac{\nu}{2}|^2=
\]
\[
\frac{1}{2}\tr_0^2\left(\eta+\zeta\right)\tr_0^2\frac{\mu}{2}
\tr_0^2\frac{\nu}{2}+(-1)^s\tr_0\left(\eta+\zeta\right)\tr_0\frac{\mu}{2}
\tr_0\frac{\nu}{2}\tr_1\frac{\mu}{2}\tr_1\frac{\nu}{2}+
\]
\[
\frac{1}{2}\tr_1^2\frac{\mu}{2}\tr_1^2\frac{\nu}{2}+\frac{1}{2}\tr_1^2\left(\eta+\zeta\right)\tr_0^2\frac{\mu}{2}
\tr_0^2\frac{\nu}{2}=
\]
\[
\frac{1}{2}(\tr_0^2\frac{\mu}{2}
\tr_0^2\frac{\nu}{2}+2(-1)^{s+1}\tr_0\frac{\mu}{2}
\tr_0\frac{\nu}{2}\tr_1\frac{\mu}{2}\tr_1\frac{\nu}{2}+\tr_1^2\frac{\mu}{2}\tr_1^2\frac{\nu}{2})+
\]
\[
(-1)^s(1+\tr_0\left(\eta+\zeta\right))\tr_0\frac{\mu}{2}
\tr_0\frac{\nu}{2}\tr_1\frac{\mu}{2}\tr_1\frac{\nu}{2}=
\]\[
\frac{1}{2}\tr_0^2\frac{\mu+(-1)^s\nu}{2}+2(-1)^s\tr_0^2\frac{\eta+\zeta}{2}P\left(\frac{\mu}{2},
\frac{\nu}{2}\right).
\]
In case $t=1$ we obtain
\[
p_{0,0}=\frac{1}{2}|
(\tr_0\left(\eta-\zeta\right)-\mathrm{i}\tr_1\left(\eta-\zeta\right))\tr_0\frac{\mu}{2}\tr_1\frac{\nu}{2}+
(-1)^s\tr_1\frac{\mu}{2}\tr_0\frac{\nu}{2}|^2=
\]
\[
\frac{1}{2}
\tr_0^2\left(\eta-\zeta\right)\tr_0^2\frac{\mu}{2}\tr_1^2\frac{\nu}{2}+
(-1)^s\tr_0\left(\eta-\zeta\right)\tr_0\frac{\mu}{2}\tr_1\frac{\nu}{2}\tr_1\frac{\mu}{2}\tr_0\frac{\nu}{2}+
\]
\[
\frac{1}{2}\tr_1^2\frac{\mu}{2}\tr_0^2\frac{\nu}{2}+
\frac{1}{2}\tr_1^2\left(\eta-\zeta\right)\tr_0^2\frac{\mu}{2}\tr_1^2\frac{\nu}{2}=
\]
\[
\frac{1}{2}\tr_0^2\frac{\mu}{2}\tr_1^2\frac{\nu}{2}+
(-1)^s\tr_0\frac{\mu}{2}\tr_1\frac{\nu}{2}\tr_1\frac{\mu}{2}\tr_0\frac{\nu}{2}+
\frac{1}{2}\tr_1^2\frac{\mu}{2}\tr_0^2\frac{\nu}{2}+
\]
\[
(-1)^{s+1}(1-\tr_0\left(\eta-\zeta\right))\tr_0\frac{\mu}{2}\tr_1\frac{\nu}{2}\tr_1\frac{\mu}{2}\tr_0\frac{\nu}{2}=
\]
\[
\frac{1}{2}\tr_1^2\frac{\mu+(-1)^s\nu}{2}+
2(-1)^{s+1}\tr_1^2\frac{\eta-\zeta}{2}P\left(\frac{\mu}{2},\frac{\nu}{2}\right).
\]
{\rm (iii)} When $t=0$ we have
\[
p_{0,1}=\frac{1}{2}|
(\tr_0\left(\eta+\zeta\right)-\mathrm{i}\tr_1\left(\eta+\zeta\right))\tr_0\frac{\mu}{2}
\tr_1\frac{\nu}{2}+
(-1)^{s+1}\tr_1\frac{\mu}{2}\tr_0\frac{\nu}{2}|^2=
\]
\[
\frac{1}{2}\tr_0^2\left(\eta+\zeta\right)\tr_0^2\frac{\mu}{2}
\tr_1^2\frac{\nu}{2}+(-1)^{s+1}\tr_0\left(\eta+\zeta\right)\tr_0\frac{\mu}{2}
\tr_1\frac{\nu}{2}\tr_1\frac{\mu}{2}\tr_0\frac{\nu}{2}+
\]
\[
\frac{1}{2}\tr_1^2\frac{\mu}{2}\tr_0^2\frac{\nu}{2}+\frac{1}{2}\tr_1^2\left(\eta+\zeta\right)\tr_0^2\frac{\mu}{2}
\tr_1^2\frac{\nu}{2}=
\]
\[
\frac{1}{2}\tr_0^2\frac{\mu}{2}
\tr_1^2\frac{\nu}{2}+(-1)^s\tr_0\frac{\mu}{2}
\tr_1\frac{\nu}{2}\tr_1\frac{\mu}{2}\tr_0\frac{\nu}{2}+\frac{1}{2}\tr_1^2\frac{\mu}{2}\tr_0^2\frac{\nu}{2}+
\]
\[
(-1)^{s+1}(1+\tr_0\left(\eta+\zeta\right))\tr_0\frac{\mu}{2}
\tr_1\frac{\nu}{2}\tr_1\frac{\mu}{2}\tr_0\frac{\nu}{2}=
\]
\[
\frac{1}{2}\tr_1^2\frac{\mu+(-1)^s\nu}{2}+
2(-1)^{s+1}\tr_0^2\frac{\eta+\zeta}{2}P\left(\frac{\mu}{2},
\frac{\nu}{2}\right).
\]
When $t=1$ we obtain
\[
p_{0,1}=\frac{1}{2}|
(\tr_0\left(\eta-\zeta\right)-\mathrm{i}\tr_1\left(\eta-\zeta\right))\tr_0\frac{\mu}{2}\tr_0\frac{\nu}{2}+
(-1)^{s+1}\tr_1\frac{\mu}{2} \tr_1\frac{\nu}{2}|^2=
\]
\[
\frac{1}{2}
\tr_0^2\left(\eta-\zeta\right)\tr_0^2\frac{\mu}{2}\tr_0^2\frac{\nu}{2}+
(-1)^{s+1}\tr_0\left(\eta-\zeta\right)\tr_0\frac{\mu}{2}\tr_0\frac{\nu}{2}\tr_1\frac{\mu}{2}\tr_1\frac{\nu}{2}+
\]
\[
\frac{1}{2}\tr_1^2\frac{\mu}{2}\tr_1^2\frac{\nu}{2}+
\frac{1}{2}\tr_1^2\left(\eta-\zeta\right)\tr_0^2\frac{\mu}{2}\tr_0^2\frac{\nu}{2}=
\]
\[
\frac{1}{2}\tr_0^2\frac{\mu}{2}\tr_0^2\frac{\nu}{2}+
(-1)^{s+1}\tr_0\frac{\mu}{2}\tr_0\frac{\nu}{2}\tr_1\frac{\mu}{2}\tr_1\frac{\nu}{2}+
\]
\[
\frac{1}{2}\tr_1^2\frac{\mu}{2}\tr_1^2\frac{\nu}{2}+
(-1)^s(1-\tr_0\left(\eta-\zeta\right))\tr_0\frac{\mu}{2}\tr_0\frac{\nu}{2}\tr_1\frac{\mu}{2}\tr_1\frac{\nu}{2}=
\]
\[
\frac{1}{2}\tr_0^2\frac{\mu+(-1)^s\nu}{2}+
2(-1)^s\tr_1^2\frac{\eta-\zeta}{2}P\left(\frac{\mu}{2},\frac{\nu}{2}\right).
\]
Thus, we proved the third equalities from parts {\rm (ii)} and
{\rm (iii)}. The fourth equalities can be obtained by the
trigonometric identity
\[
P\left(\frac{\mu}{2},\frac{\nu}{2}\right)=
\frac{1}{4}(-1)^{s+t+1}(\tr_t^2\frac{\mu+(-1)^s\nu}{2}-\tr_t^2\frac{\mu+(-1)^{s+1}\nu}{2}).
\]

\end{proof}

\section{The Noise}

\label{10}

Here we follow~\cite[Section 5]{[5]} and~\cite[Section 3]{[10]}
with $\alpha=\beta=\frac{1}{2}$, $A=(\mathcal{A}_{\mu,\eta}=1)$,
$B=(\mathcal{B}_{\nu,\zeta}=1)$.

The joint experiment of the binary trials
$\mathfrak{A}_{\mu,\eta}=
(\mathcal{A}_{\mu,\eta}=1)\cup(\mathcal{A}_{\mu,\eta}=-1)$ and
$\mathfrak{B}_{\nu,\zeta}=
(\mathcal{B}_{\nu,\zeta}=1)\cup(\mathcal{B}_{\nu,\zeta}=-1)$
(see~\cite[Part I, Section 6]{[15]}) produces the probability
distribution $\xi_1=p_{k,k},\xi_2=p_{k,k+1},
\xi_3=p_{k,k+1},\xi_4=p_{k,k}, k\in\mathbb{Z}_2$, (that is, the
probability assignment of the sample space
$S(\psi_{s,t};\mathcal{A},\mathcal{B})$), see Theorem~\ref{5.5.1}
). In turn, we obtain Boltzmann-Shannon entropy function
$E(\theta)=
-2\theta\ln\theta-2(\frac{1}{2}-\theta)\ln(\frac{1}{2}-\theta)$,
where $\theta=\xi_1=p_{k,k}$, and the corresponding degree of
dependence function $e(\theta)$. In accord with Section 3
of~\cite{[10]}, the two events $\mathcal{A}_{\mu,\eta}=\lambda_k$
and $\mathcal{B}_{\nu,\zeta}=\lambda_\ell$ are independent (that
is, $e(\theta)=0$), exactly when the corresponding binary trials
$\mathfrak{A}_{\mu,\eta}$ and $\mathfrak{B}_{\nu,\zeta}$ are
informationally independent. In this case we also say that the
above events are \emph{informationally independent}. Since their
probabilities are both $\frac{1}{2}$, we accomplish the equality
$\theta=\frac{1}{4}$ as an equivalent condition for independence.
In accord with Theorem~\ref{5.5.1}, if $\psi=\psi_{s,t}$ for some
$s,t\in\mathbb{Z}_2$, then the equation
\[
\frac{1}{2}\tr_t^2\frac{\mu+(-1)^s\nu}{2}+2(-1)^{s+t}\tr_t^2\frac{\eta+(-1)^t\zeta}{2}P\left(\frac{\mu}{2},
\frac{\nu}{2}\right)=\frac{1}{4},
\]
or, equivalently, the equation
\[
\frac{1}{2}\tr_t^2\frac{\mu+(-1)^{s+1}\nu}{2}+
2(-1)^{s+t+1}\tr_{t+1}^2\frac{\eta+(-1)^t\zeta}{2}P\left(\frac{\mu}{2},
\frac{\nu}{2}\right)=\frac{1}{4},
\]
is a necessary and sufficient condition for the events
$\mathcal{A}_{\mu,\eta}=\lambda_k$ and
$\mathcal{B}_{\nu,\zeta}=\lambda_\ell$ to be informationally
idependent.

\section{Special Types of Self-Adjoined Operators with Spectre $\{1,-1\}$}

\label{15}

In this Section we apply conditions of informational independence
from Section~\ref{10} in some particular cases where they have
simple form and explicit solutions.

\subsection{Operators with Direction\\ in Coordinate Plane $x=0$}

\label{15.1}

Here we consider self-adjoined operators of the form
\[
A_{\mu,\frac{\pi}{2}}=\left(
\begin{array}{ccccc}
\cos\mu &-\mathrm{i}\sin\mu \\
\mathrm{i}\sin\mu & -\cos\mu \\
 \end{array}
\right),
\]
where $\mu\in [0,\pi]$. In particular,
$A_{0,\frac{\pi}{2}}=\sigma_3$ and
$A_{\frac{\pi}{2},\frac{\pi}{2}}=\sigma_2$. The two events
$\mathcal{A}_{\mu,\eta}=\lambda_k$ and
$\mathcal{B}_{\nu,\zeta}=\lambda_\ell$ are informationally
independent if and only if
$\frac{1}{2}\tr_t^2\frac{\mu+(-1)^s\nu}{2}=\frac{1}{4}$.
Equivalently, $|\mu+(-1)^s\nu|=\frac{\pi}{2}$ or
$|\mu+(-1)^s\nu|=\frac{3\pi}{2}$. Hence the outcomes of
measurement of observables $\mathcal{A}_{\mu,\eta}$ and
$\mathcal{B}_{\nu,\zeta}$ are informationally independent
precisely in case $\mu+\nu=\frac{\pi}{2}$ or
$\mu+\nu=\frac{3\pi}{2}$ when $\psi=\psi_{0,t}$,
$t\in\mathbb{Z}_2$, and precisely in case
$|\mu-\nu|=\frac{\pi}{2}$ when $\psi=\psi_{1,t}$,
$t\in\mathbb{Z}_2$.

\subsection{Operators with Direction \\in Coordinate Plane $y=0$}

\label{15.5}

Now, we consider the self-adjoined operators  of the form
\[
A_{\mu,0}=\left(
\begin{array}{ccccc}
\cos\mu &\sin\mu \\
\sin\mu & -\cos\mu \\
 \end{array}
\right)
\]
where $\mu\in [0,\pi]$. We have $A_{0,0}=\sigma_3$,
$A_{\frac{\pi}{2},0}=\sigma_1$, and $A_{\frac{\pi}{4},0}=H$.

The two events $\mathcal{A}_{\mu,\eta}=\lambda_k$ and
$\mathcal{B}_{\nu,\zeta}=\lambda_\ell$ are informationally
independent if and only if
$\frac{1}{2}\tr_t^2\frac{\mu+(-1)^{s+t+1}\nu}{2}=\frac{1}{4}$.
Equivalently, $|\mu+(-1)^{s+t+1}\nu|=\frac{\pi}{2}$. Therefore the
results of measurement of observables $\mathcal{A}_{\mu,\eta}$ and
$\mathcal{B}_{\nu,\zeta}$ are informationally independent exactly
in case $\mu+\nu=\frac{\pi}{2}$ or $\mu+\nu=\frac{3\pi}{2}$ when
$\psi=\psi_{s,s+1}$, and exactly in case $|\mu-\nu|=\frac{\pi}{2}$
when $\psi=\psi_{s,s}$, $s\in\mathbb{Z}_2$.

\begin{remark}\label{15.5.1} {\rm The case $t=s=1$ is discussed also in
~\cite{[10]}. }
\end{remark}

\subsection{Operators with Direction \\ in Coordinate Plane $z=0$}

\label{15.10}

Finally, we consider self-adjoined operators  of the form
\[
A_{\frac{\pi}{2},\eta}=\left(
\begin{array}{ccccc}
0 &e^{-\mathrm{i}\eta} \\
e^{\mathrm{i}\eta} & 0 \\
 \end{array}
\right),
\]
where $\eta\in [0,2\pi]$. We have $A_{\frac{\pi}{2},0}=\sigma_1$
and $A_{\frac{\pi}{2},\frac{\pi}{2}}=\sigma_2$.

The two events $\mathcal{A}_{\mu,\eta}=\lambda_k$ and
$\mathcal{B}_{\nu,\zeta}=\lambda_\ell$ are informationally
independent if and only if
$\frac{1}{2}\tr_s^2\frac{\eta+(-1)^t\zeta}{2}=\frac{1}{4}$.
Equivalently, $|\eta+(-1)^t\zeta|=\frac{\pi}{2}$. Therefore, the
measurements performed by the observables $\mathcal{A}_{\mu,\eta}$
and $\mathcal{B}_{\nu,\zeta}$ are informationally independent
precisely in case
$\eta+\zeta\in\{\frac{\pi}{2},\frac{3\pi}{2},\frac{5\pi}{2},\frac{7\pi}{2}\}$,
when $\psi=\psi_{s,0}$, and precisely in case
$|\eta-\zeta|=\frac{\pi}{2}$ or $|\eta-\zeta|=\frac{3\pi}{2}$,
when $\psi=\psi_{s,1}$, $s\in\mathbb{Z}_2$.

\section*{Acknowledgements}

I would like to thank administration of the Institute of
Mathematics and Informatics at the Bulgarian Academy of Sciences
for creating safe working environment.

\end{document}